\documentclass[aps,prl,preprintnumbers,twocolumn,groupedaddress,nofootinbib]{revtex4}

\usepackage{latexsym}
\usepackage{amsmath}
\usepackage{amsfonts}
\usepackage{amsthm}
\usepackage{graphicx}

\usepackage{color}

\usepackage[latin1,utf8]{inputenc}
\usepackage[T1]{fontenc}

\allowdisplaybreaks

\newcommand{\bq}{\begin{eqnarray}}
\newcommand{\eq}{\end{eqnarray}}
\newcommand{\eps}{\varepsilon}

\theoremstyle{plain}

\newtheorem{theoremcounter}{}[]
\newtheorem{theorem}[theoremcounter]{Theorem}

\newcommand{\arxivdate}{February 6, 2019}

\begin{document}

\preprint{MITP/19-009}
\title{\boldmath{Causality and loop-tree duality at higher loops}}

\author{Robert Runkel, Zolt\'an Sz\H{o}r, Juan Pablo Vesga and Stefan Weinzierl}
\affiliation{PRISMA Cluster of Excellence, Institut f{\"u}r Physik, Johannes Gutenberg-Universit\"at Mainz, D-55099 Mainz, Germany}

\date{\arxivdate}

\begin{abstract}
We relate a $l$-loop Feynman integral to a sum of phase space integrals, where the integrands
are determined by the spanning trees of the original $l$-loop graph.
Causality requires that the propagators of the trees have a modified $i\delta$-prescription and 
we present a simple formula for the correct $i\delta$-prescription.
\end{abstract}

\maketitle

\section{Introduction}
\label{sec:intro}

Relating loop integrals to trees goes back to Feynman \cite{Feynman:1963ax}.
The Feynman tree theorem allows us to relate a $l$-loop Feynman integral with $N$ internal propagators
to a $l$-fold phase space integral with a number of cuts $N_{\mathrm{cut}}$ 
on the original integrand, with $N_{\mathrm{cut}}$ ranging from $l$ to $N$.
One could argue that a better name for this theorem would be the Feynman forest theorem, 
as in general (i.e. for $N_{\mathrm{cut}}>l$) the integrand
corresponds to a set of trees, i.e. a forest.
This is not very convenient: If more than $l$ cuts are present, each additional cut 
imposes a non-trivial constraint on the phase space integration.

What we would like to have is a formula which relates a $l$-loop Feynman integral to a $l$-fold phase space integral
without any additional constraints.
The integrand of the phase space integral then corresponds to a tree, not a forest, and is obtained from the original
integrand by exactly $l$ cuts.
For one-loop integrals this was achieved in \cite{Catani:2008xa}.
An important result of this paper was the statement that the uncut propagators have a modified
$i \delta$-prescription. 
The usual $i \delta$-prescription for a Feynman propagator is
\bq
\label{Feynman_propagator}
 \frac{i}{k_j^2-m_j^2+i\delta},
\eq
where $\delta>0$ is an infinitesimal small quantity.
The modified $i \delta$-prescription is a consequence of causality.
We call propagators with a modified $i \delta$-prescription {\it dual propagators}.

In this letter we present the generalisation to an arbitrary loop number $l$.
A comment is in order: A generalisation of loop-tree duality to two loops and beyond has already
been considered in \cite{Bierenbaum:2010cy}.
However, the final formulae presented there are not particular elegant and involve a mixture of Feynman propagators
and one-loop dual propagators.
Our result is more aesthetic: All uncut propagators are dual propagators, with a simple dual $i \delta$-prescription,
which reduces in the one-loop case to the one of Catani et al. \cite{Catani:2008xa}.
A dual propagator is of the form
\bq
 \frac{i}{k_j^2-m_j^2+i s_j\left(\sigma\right) \delta}.
\eq
Only the sign of the function $s_j(\sigma)$ is relevant. The function $s_j(\sigma)$ depends on the energy $E_j$
and the energies of the cut propagators $E_{\sigma_1}$, ..., $E_{\sigma_l}$ and will be given in eq.~(\ref{dual_delta_I_prescription}) 
or alternatively in eq.~(\ref{dual_delta_I_prescription_covariant}) below.

Let us also briefly comment on the difference of the loop-tree duality approach 
with the so-called $Q$-cut approach \cite{Baadsgaard:2015twa}:
The latter involves propagators linear in the loop momenta, 
where the information due to the infinitesimal imaginary part is lost if all propagators have the same small imaginary part as in eq.~(\ref{Feynman_propagator}).
The $i\delta$-prescription is restored by introducing different infinitesimal small imaginary parts for the
internal propagators and averaging over all possible relative orderings.
In our approach, only propagators quadratic in the loop momenta occur.
Furthermore, the modified $i \delta$-prescription of the dual propagators follows directly from
$i \delta$-prescription of the Feynman propagators.
Our approach applies to massless and massive particles.

\section{Notation}
\label{sec:notation}

Let $\Gamma$ be a Feynman graph with $l$ loops, $n$ external lines and $N$ internal edges.
We denote by $E_\Gamma=\{e_1,...,e_N\}$ the set of internal edges.
A spanning tree for the graph $\Gamma$ is a sub-graph $T$ of $\Gamma$, 
which contains all the vertices of $\Gamma$ and is a connected tree graph \cite{Bogner:2010kv}.
If $T$ is a spanning tree for $\Gamma$, 
then it can be obtained from $\Gamma$ by deleting $l$ internal edges, say $\{e_{\sigma_1},...,e_{\sigma_l}\}$.
We denote by $\sigma=\{\sigma_1,...,\sigma_l\}$ the set of indices of the deleted edges
and by ${\mathcal C}_\Gamma$ the set of all such sets.
Thus $|{\mathcal C}_\Gamma|$ gives the number of spanning trees for the graph $\Gamma$.

Each $\sigma \in {\mathcal C}_\Gamma$ defines also a cut graph $T_{\mathrm{cut}}$, obtained by cutting
each of the $l$ internal edges $e_{\sigma_j}$ into two half-edges. The $2l$ half-edges become external lines of $T_{\mathrm{cut}}$.
The graph $T_{\mathrm{cut}}$ is a tree graph with $n+2l$ external lines.

We denote the external momenta of the graph $\Gamma$ by $p_1$, ..., $p_n$ and the internal momenta by $k_1$, ..., $k_N$.
We will assume that the internal momenta have been labelled such that the first $l$ internal momenta $k_1$, ..., $k_l$
form a basis of independent loop momenta.
For each internal edge we set
\bq
 D_j & = & k_j^2 - m_j^2 + i \delta, 
 \;\;\;\;\;\;
 e_j \in E_\Gamma.
\eq
We will assume that $D_i \neq D_j$  for $i \neq j$.
Otherwise we consider a reduced graph $\Gamma'$ with edge $e_j$ contracted and a higher power of the propagator
associated to edge $e_i$.

For a function $f$ depending on a $D$-dimensional momentum variable $k=(E,\vec{k})$,
where the vector $\vec{k}$ is $(D-1)$-dimensional,
we either write 
$f(k)$ or $f(E,\vec{k})$.
We would like to integrate the function $f$ over the hyperboloid $k^2=m^2$.
The quantities
\bq
 \int \frac{d^{D-1}k}{\left(2\pi\right)^{D-1} 2 \sqrt{\vec{k}^2+m^2}}
 \; 
 f\left(\pm\sqrt{\vec{k}^2+m^2},\vec{k}\right)
\eq
give the integrals over the forward hyperboloid and the backward hyperboloid, respectively.
As a short hand notation we set
\bq
\lefteqn{
 \int \frac{d^{D-1}k}{\left(2\pi\right)^{D-1} 2 \sqrt{\vec{k}^2+m^2}} 
 \;
 f\left(k\right)
 = 
 \int \frac{d^{D-1}k}{\left(2\pi\right)^{D-1} 2 \sqrt{\vec{k}^2+m^2}} 
 } & &
 \nonumber \\
 & &
 \hspace*{6mm}
 \times
 \left[ 
 f\left(\sqrt{\vec{k}^2+m^2},\vec{k}\right)
 +
 f\left(-\sqrt{\vec{k}^2+m^2},\vec{k}\right)
 \right]
 \hspace*{4mm}
\eq
for the integral over the forward and the backward hyperboloid.

\section{Loop-tree duality}
\label{sec:duality}

Let $P_\Gamma$ be a polynomial in the loop momenta.
We consider
\bq
 I & = &
 \int \left( \prod\limits_{j=1}^l \frac{d^Dk_j}{\left(2\pi\right)^D} \right)
 \frac{P_\Gamma}{\prod\limits_{e_j \in E_\Gamma} D_j^{\nu_j}}.
\eq
We split each loop integration into an integration over the energy and the spatial components of the loop momentum:
\bq
 I 
 & = &
 \int \left( \prod\limits_{j=1}^l \frac{d^{D-1}k_j}{\left(2\pi\right)^{D-1}} \right)
 \frac{1}{\left(2\pi\right)^l}
 \int 
 \frac{P_\Gamma dE_1 \wedge ... \wedge dE_l}{\prod\limits_{e_j \in E_\Gamma} D_j^{\nu_j}}.
 \hspace*{4mm}
\eq
We perform the energy integrations with the help of the residue theorem.
Let us assume that the polynomial $P_\Gamma$ is such that 
all energy integrations over half circles at infinity vanish.
This assumption is always satisfied for scalar integrals where 
$P_\Gamma=1$.
If this assumption is not met, we may enforce it by subtracting local ultraviolet counterterms from the integrand \cite{Becker:2010ng,Becker:2012aq}.

Let ${\mathcal E} \subset {\mathbb C}^l$ be the set of points $E = (E_1,...,E_l)$, where $l$ internal propagators go on-shell
and removing the corresponding edges gives a spanning tree.
We have
\bq
 \left| {\mathcal E} \right|
 & = &
 2^l \left| {\mathcal C}_\Gamma \right|.
\eq
This number is easily obtained from the number of spanning trees 
and the $2^l$ solutions per spanning tree.
For generic values of $p_1, ..., p_n$ and $\vec{k}_1, ..., \vec{k}_l$ the points of ${\mathcal E}$ are distinct.
Points in ${\mathcal E}$ coincide if in addition to the cut propagators 
one or more uncut propagators go on-shell.
We distinguish the cases of a pinch singularity and a non-pinch singularity.
For a non-pinch singularity we may deform the integration contour for the spatial variables $\vec{k}_1$, ..., $\vec{k}_l$
into the complex domain. 
The modified $i \delta$-prescription given in eq.~(\ref{dual_delta_I_prescription}) tells us 
in which direction we should deform. This is exactly the raison d'\^etre for the present article.
For a pinch singularity we have an infrared singularity. This singularity is 
either regulated by dimensional regularisation or cancelled in the combination with real contributions according to the
Kinoshita-Lee-Nauenberg theorem \cite{Kinoshita:1962ur,Lee:1964is}.

Let $\sigma \in {\mathcal C}_\Gamma$ be a set of indices defining a spanning tree.
For each cut edge we choose an orientation 
and we may take the $l$ independent loop momenta to be the loop momenta flowing
through the edges $e_{\sigma_1}, ..., e_{\sigma_l}$ with the chosen orientation.
Let
\bq
 E_\sigma^{(\alpha)} & = & \left(E_{\sigma_1}^{(\alpha)},...,E_{\sigma_l}^{(\alpha)}\right)
\eq
be a solution to
\bq
 D_{\sigma_1}
 \;\; = \;\;
 ...
 \;\; = \;\;
 D_{\sigma_l}
 \;\; = \;\;
 0.
\eq
In total there are $2^l$ solutions $E_\sigma^{(1)}$, ..., $E_\sigma^{(2^l)}$,
given by
\bq
 \left( \pm \sqrt{ \vec{k}_{\sigma_1}^2 + m_{\sigma_1}^2 - i \delta}, ..., \pm \sqrt{ \vec{k}_{\sigma_l}^2 + m_{\sigma_l}^2 - i \delta} \right).
\eq
Let us denote by $n_\sigma^{(\alpha)}$ the number of times the negative root $-\sqrt{...}$ occurs in
$E_\sigma^{(\alpha)}$.
We set
\bq
\label{def_f}
 f & = &
  \frac{P_\Gamma}{\prod\limits_{e_j \in E_\Gamma} D_j^{\nu_j}}.
\eq
We define the local residue \cite{Griffiths:book} at $E_\sigma^{(\alpha)}$ by
\bq
\label{local_residue}
 \mathrm{res}\left(f,E_\sigma^{(\alpha)}\right)
 & = &
 \frac{1}{\left(2\pi i\right)^l}
 \oint\limits_{\gamma_\eps} f dE_1 \wedge ... \wedge dE_l.
\eq
The integration in eq.~(\ref{local_residue}) is around a small $l$-torus
\bq
 \gamma_\eps & = &
 \left\{
   \left( E_1, ..., E_l \right) \in {\mathbb C}^l | \left| D_{\sigma_i} \right| = \eps
 \right\},
\eq
encircling $E_\sigma^{(\alpha)}$ with orientation
\bq
 d \arg D_{\sigma_1} \wedge d \arg D_{\sigma_2} \wedge ... \wedge d \arg D_{\sigma_l} \ge 0.
\eq
We consider the weighted sum of residues
\bq
\label{sum_residues}
 \sum\limits_{\alpha=1}^{2^l}
 \left(-1\right)^{n_\sigma^{(\alpha)}}
 S_{\sigma \alpha} \;
 \mathrm{res}\left(f,E_\sigma^{(\alpha)}\right),
\eq
where $\left(-1\right)^{n_\sigma^{(\alpha)}} S_{\sigma \alpha}$ is a weight factor depending on $\sigma$ and $\alpha$.
Let us make one remark: Eq.(\ref{sum_residues}) is not a global residue for the ideal $\langle D_{\sigma_1}^{\nu_{\sigma_1}}, ..., D_{\sigma_l}^{\nu_{\sigma_l}} \rangle$,
due to the additional factor $(-1)^{n_\sigma^{(\alpha)}} S_{\sigma \alpha}$.
The standard definition of the global residue for $\langle D_{\sigma_1}^{\nu_{\sigma_1}}, ..., D_{\sigma_l}^{\nu_{\sigma_l}} \rangle$ is just the sum over the $2^l$ local residues,
without any weight factors.
This sum vanishes, whereas the sum in eq.~(\ref{sum_residues}) does in general not.
\begin{theorem}
\label{theorem_arbitrary_powers}
With $f$ as in eq.~(\ref{def_f}) we have
\bq
\label{contour_integration}
\lefteqn{
 \frac{1}{\left(2\pi\right)^l}
 \int f dE_1 \wedge ... \wedge dE_l
 = } & &
 \nonumber \\
 & &
 \left(-i\right)^l
 \sum\limits_{\sigma \in {\mathcal C}_\Gamma}
 \sum\limits_{\alpha=1}^{2^l}
 S_{\sigma \alpha} \;
 \left(-1\right)^{n_\sigma^{(\alpha)}}
 \mathrm{res}\left(f,E_\sigma^{(\alpha)}\right),
\eq
where the contour of integration on the left-hand side is along the real axes separating 
the poles at $+\sqrt{...}$ from
the poles at $-\sqrt{...}$. 
\end{theorem}
\begin{proof}
We specify a set of integration variables by $\tilde{\sigma} \in {\mathcal C}_\Gamma$
and an order in which the integrations are performed by $\tilde{\pi} \in S_l$.
We assume that the integration over $k_{\tilde{\sigma}_{\tilde{\pi_1}}}$ is performed first, followed
by the integration over $k_{\tilde{\sigma}_{\tilde{\pi_2}}}$, etc..
In order to keep the indexing to a minimum we introduce the ordered set
$\tilde{k}=(\tilde{k}_1,...,\tilde{k}_l)=(k_{\tilde{\sigma}_{\tilde{\pi_1}}},...,k_{\tilde{\sigma}_{\tilde{\pi_l}}})$.
Let further $\tilde{\alpha}=(\Gamma_1,...,\Gamma_l)$ be the ordered set of winding numbers.
For a cut specified by $\sigma \in {\mathcal C}_\Gamma$ we denote by $\pi \in S_l$ the order in which the cuts
are taken, e.g. the cut of the edge $e_{\sigma_{\pi_1}}$ is taken in the first integration, followed
by the the cut of the edge $e_{\sigma_{\pi_2}}$, etc..
Again, in order to keep the indexing to a minimum we introduce the ordered set
$\hat{k}=(\hat{k}_1,...,\hat{k}_l) = (k_{\sigma_{\pi_1}}, ..., k_{\sigma_{\pi_l}})$.
We denote by $\alpha=(\lambda_1,...,\lambda_l)$ the signs of the energies for the cut under consideration.
$\lambda_j=1$ means that we consider the residue with positive energy with respect to the chosen
orientation of the edge $e_{\sigma_{\pi_j}}$.
$\hat{k}$ and $\tilde{k}$ are both bases of independent loop momenta, hence they are related by
\bq
 \hat{k}_i & = & \sum\limits_{j=1}^l \Sigma_{ij} \tilde{k}_j + q_i,
\eq
with $\Sigma_{ij} \in \{-1,0,1\}$ and $q_i$ depending only on the external momenta.
This defines the $l \times l$-signature matrix $\Sigma$.
We denote by $\Sigma^{(j)}$ the $j \times j$-matrix obtained from $\Sigma$ by deleting the rows and columns
$(j+1), ..., l$.
In order to compute the residues we may temporarily assume that the imaginary parts of all internal masses are large and strongly ordered.
The final result will not depend on this assumption.
After performing the contour integrations we may remove this assumption and analytically continue to any desired (complex) kinematics.
With these specifications one obtains
\bq
\label{corrected_result}
\lefteqn{
 \frac{1}{\left(2\pi i\right)^l}
 \int f dE_1 \wedge ... \wedge dE_l
 = 
} & &
 \\
 & &
 \sum\limits_{\sigma \in {\mathcal C}_\Gamma}
 \sum\limits_{\pi \in S_l}
 \sum\limits_{\alpha \in \{1,-1\}^l}
 C^{\tilde{\sigma} \tilde{\pi} \tilde{\alpha}}_{\sigma \pi \alpha}
 \;
 \mathrm{res}\left(f,E_\sigma^{(\alpha)}\right),
 \nonumber
\eq
where
$C^{\tilde{\sigma} \tilde{\pi} \tilde{\alpha}}_{\sigma \pi \alpha}$ is given by
\bq
 C^{\tilde{\sigma} \tilde{\pi} \tilde{\alpha}}_{\sigma \pi \alpha}
 & = &
 \prod\limits_{i=1}^l \Delta^{(i)}.
\eq
$\Delta^{(i)}$ is zero if $\det \Sigma^{(i)} = 0$.
Otherwise we let $\Pi^{(i)}$ be the inverse matrix of $\Sigma^{(i)}$.
The quantity $\Delta^{(i)}$ is then given by
\bq
 \Delta^{(i)}
 & = &
 \Gamma_i \Pi^{(i)}_{ii}
 \;
 \theta\left( \Gamma_i \mathrm{Im}\left( \sum\limits_{j=1}^i \Pi^{(i)}_{ij} \lambda_j m_{\sigma_{\pi_j}} \right) \right).
\eq
The quantities $\Delta^{(i)}$ are computed with a chosen strong ordering of the imaginary parts of the internal masses.
The quantity $C^{\tilde{\sigma} \tilde{\pi} \tilde{\alpha}}_{\sigma \pi \alpha}$ is independent of this choice.
Eq.~(\ref{corrected_result}) generalises eq.~(2) of \cite{Capatti:2019ypt} to complex external kinematics.

One may now sum over $\pi$ and average over $\tilde{\sigma}$, $\tilde{\alpha}$, $\tilde{\pi}$ 
in a suitable way.
We do this as follows:
We group the internal propagators $D_j$ into chains \cite{Kinoshita:1962ur}.
Two propagators belong to the same chain, if their momenta differ only by a linear combination 
of the external momenta. We denote by $n^{\mathrm{chain}}(j)$ the number of propagators in the chain of $D_j$.
We set
\bq
 N^{\mathrm{chain}}\left(\sigma\right)
 & = &
 \prod\limits_{j=1}^l n^{\mathrm{chain}}\left(\sigma_j\right).
\eq
To each graph $\Gamma$ we associate a new graph $\Gamma^{\mathrm{chain}}$ called the chain graph by deleting all external lines and by choosing one propagator for each chain
as a representative.
We denote by $|{\mathcal C}_{\Gamma^{\mathrm{chain}}}|$ the number of spanning trees of the chain graph.
We then perform a weighted average, where each term is weighted by $1/N^{\mathrm{chain}}(\sigma)$.
We obtain
\bq
\lefteqn{
 \frac{1}{\left(2\pi\right)^l}
 \int f dE_1 \wedge ... \wedge dE_l
 = } & &
 \nonumber \\
 & &
 \left(-i\right)^l
 \sum\limits_{\sigma \in {\mathcal C}_\Gamma}
 \sum\limits_{\alpha=1}^{2^l}
 S_{\sigma \alpha}
 \left(-1\right)^{n_\sigma^{(\alpha)}}
 \mathrm{res}\left(f,E_\sigma^{(\alpha)}\right),
\eq
with
\bq
\lefteqn{
 S_{\sigma \alpha}
 = } & &
 \\ 
 & &
 \frac{\left(-1\right)^{l+n_\sigma^{(\alpha)}}}{2^l l! \left|{\mathcal C}_{\Gamma^{\mathrm{chain}}}\right|}
 \sum\limits_{\pi \in S_l}
 \sum\limits_{\tilde{\sigma} \in {\mathcal C}_\Gamma}
 \sum\limits_{\tilde{\pi} \in S_l}
 \sum\limits_{\tilde{\alpha} \in \{1,-1\}^l}
 \frac{C^{\tilde{\sigma} \tilde{\pi} \tilde{\alpha}}_{\sigma \pi \alpha}}{N^{\mathrm{chain}}\left(\sigma\right)}.
 \nonumber
\eq
This defines the $S_{\sigma \alpha}$.
\end{proof}

The factor $S_{\sigma \alpha}$ equals $1/2$ for all one-loop graphs, 
it equals $1/4$ for all two-loop graphs whose underlying chain graph is a product of two one-loop tadpoles,
while it equals
\bq
\label{banana_graphs}
 \frac{1}{\left(l+1\right)} \frac{1}{\left(\begin{array}{c} l \\ n_\sigma^{(\alpha)} \\ \end{array} \right)}
\eq
for all two-loop graphs whose underlying chain graph is the sunrise graph, 
if the orientation of the cut lines is chosen the same across the cut.
This agrees with \cite{CaronHuot:2010zt}.
This covers all two-loop graphs.
Eq.~(\ref{banana_graphs}) generalises to all higher loop graphs, whose underlying chain graph is a banana graph.

Let us now specialise to the case $\nu_1=...=\nu_N=1$ and let us work out the $i \delta$-prescription for the uncut propagators.
\begin{theorem}
\label{theorem_unit_powers}
Let us assume again that $P_\Gamma$ is a polynomial in the loop momenta such 
that all energy integrations over half cycles at infinity vanish.
If all propagators occur to power one, we have
\bq
\lefteqn{
 \int \left( \prod\limits_{j=1}^l \frac{d^Dk_j}{\left(2\pi\right)^D} \right)
 \frac{P_\Gamma}{\prod\limits_{e_j \in E_\Gamma} \left( k_j^2 - m_j^2 + i \delta \right)}
 = 
 } & &
 \nonumber \\
 & &
 \left(-i\right)^l
 \sum\limits_{\sigma \in {\mathcal C}_\Gamma}
 \int \left( \prod\limits_{j=1}^l \frac{d^{D-1}k_{\sigma_j}}{\left(2\pi\right)^{D-1} 2 \sqrt{\vec{k}_{\sigma_j}^2 + m_{\sigma_j}^2 }} \right)
 S_{\sigma \alpha} 
 \nonumber \\
 & &
 \frac{P_\Gamma}{\prod\limits_{j \notin \sigma } \left( k_j^2 - m_j^2 + i s_j\left(\sigma\right) \delta \right)},
\eq
where $s_j(\sigma)$ is defined as
\bq
\label{dual_delta_I_prescription}
 s_j\left(\sigma\right)
 & = &
 \frac{E_j}{E_\parallel}
\eq
and $E_\parallel$ is defined as follows:
The set $\sigma = \{\sigma_1,...,\sigma_l\} \in {\mathcal C}_\Gamma$ defines 
a tree $T_{\mathrm{cut}}$ obtained from the graph $\Gamma$ by cutting the internal edges 
$C_\sigma = \{e_{\sigma_1},...,e_{\sigma_l}\}$.
Cutting in addition the edge $e_j \in E_\Gamma \backslash C_\sigma$ 
will give a two-forest $(T_1,T_2)$.
We orient the external momenta of $T_1$ such that all momenta are outgoing.
Let $\pi$ be the set of indices corresponding to the external edges of $T_1$ which come from cutting the edges
$C_\sigma$ of the graph $\Gamma$.
The set $\pi$ may contain an index twice, this is the case if both half-edges of a cut edge belong to $T_1$.
Then define $E_\parallel$ by
\bq
 \frac{1}{E_\parallel}
 & = &
 \sum\limits_{a \in \{j\} \cup \pi} \frac{1}{E_a}.
\eq
\end{theorem}
\begin{proof}
Theorem~\ref{theorem_unit_powers} is a specialisation of theorem~\ref{theorem_arbitrary_powers} to the case where the integrand has only single poles.
The calculation of the residues yields
\bq
 \mathrm{res}\left(\frac{1}{D_{\sigma_j}},\sqrt{\vec{k}_{\sigma_j}^2+m_{\sigma_j}^2-i\delta}\right)
 & = &
 \frac{1}{2 \sqrt{\vec{k}_{\sigma_j}^2+m_{\sigma_j}^2}},
 \nonumber \\
 - \mathrm{res}\left(\frac{1}{D_{\sigma_j}},-\sqrt{\vec{k}_{\sigma_j}^2+m_{\sigma_j}^2-i\delta}\right)
 & = &
 \frac{1}{2 \sqrt{\vec{k}_{\sigma_j}^2+m_{\sigma_j}^2}},
 \nonumber
\eq
where we neglected on the right-hand side the infinitesimal small imaginary part.
It remains to work out the sign of the imaginary part of the uncut propagators. 
Let us consider $D_j$ and with the notation as above the tree $T_1$. 
The external edges of $T_1$ are given by $e_j$, the set $E_{\mathrm{cut}}=\{e_{\pi_1},e_{\pi_2},...\}$ 
and possibly a subset $E_{\mathrm{ext}}$
of the external edges $\{1,...,n\}$ of the original graph $\Gamma$.
Energy conservation relates $E_j$ to minus the sum of the energies of all other external particles
of the tree $T_1$.
The energies corresponding to the edges from $E_{\mathrm{ext}}$ are real,
the energies corresponding to the edges from $E_{\mathrm{cut}}$ have an infinitesimal small imaginary part.
Taylor expansion to first order in $\delta$ gives
\bq
 E_{\pi_a} & = & \pm \sqrt{\vec{k}_{\pi_a}^2 + m_{\pi_a}^2 - i \delta}
 \\
 & = &
 \pm \sqrt{\vec{k}_{\pi_a}^2 + m_{\pi_a}^2} \mp \frac{1}{2} \frac{i \delta}{\sqrt{\vec{k}_{\pi_a}^2 + m_{\pi_a}^2}} 
 + {\mathcal O}\left(\delta^2\right).
 \nonumber
\eq
By a slight abuse of notation we denote the ${\mathcal O}(\delta^0)$-term again by $E_{\pi_a}$.
Thus, the replacement
\bq
 E_{\pi_a}
 & \rightarrow &
 E_{\pi_a}
 - \frac{i \delta}{2 E_{\pi_a}}
\eq
makes the infinitesimal imaginary part explicit.
Let us now look at $D_j$ and expand to first order in $\delta$:
\bq
 D_j 
 & = & 
 k_j^2 - m_j^2 + i \delta
 \;\; = \;\; 
 E_j^2 - \vec{k}_j^2 - m_j^2 + i \delta
 \\
 & = &
 \left(\sum\limits_{a \in \pi} E_a + \sum\limits_{a \in E_{\mathrm{ext}}} E_a^{\mathrm{ext}} \right)^2 - \vec{k}_j^2 - m_j^2 + i \delta
 \nonumber \\
 & \rightarrow &
 k_j^2 - m_j^2 
 + \left( 1 + E_j \sum\limits_{a \in \pi} \frac{1}{E_a} \right) i \delta 
 + {\mathcal O}\left(\delta^2\right).
 \nonumber
\eq
For the ${\mathcal O}(\delta)$-term we have
\bq
 1 + E_j \sum\limits_{a \in \pi} \frac{1}{E_a} 
 & = &
 E_j \left( \sum\limits_{a \in \{j\} \cup \pi} \frac{1}{E_a} \right)
 \;\; = \;\;
 \frac{E_j}{E_\parallel}.
 \nonumber
\eq
Although we singled out the tree $T_1$ from the two-forest $(T_1,T_2)$ it is easily checked that the definition
of $s_j(\sigma)$ is invariant
under the exchange $T_1 \leftrightarrow T_2$.
\end{proof}
Theorem~\ref{theorem_unit_powers} is the main result of this letter.
It allows us to express a Feynman integral with no raised propagators as a sum of phase space integrals.
Each phase space integral corresponds to a spanning tree of the original graph.
The integrand of each phase space integral corresponds to a cut graph, where exactly $l$ internal propagators
have been cut and the remaining $(N-l)$ internal propagators have a modified $i\delta$-prescription given
by eq.~(\ref{dual_delta_I_prescription}).
Theorem~\ref{theorem_unit_powers} is the specialisation of theorem~\ref{theorem_arbitrary_powers}
to Feynman integrals with no raised propagators. 
For Feynman integrals with raised propagators we may still use theorem~\ref{theorem_arbitrary_powers}.
The only change is that the computation of the residues is more involved.
Residues of Feynman integrands with raised propagators have been considered in \cite{Bierenbaum:2012th}.

Although we defined the modified $i\delta$-prescription in terms of energies in a specific Lorentz frame, 
we may easily formulate it in a Lorentz-covariant way: Let $\eta$ be a Lorentz vector with
$\eta_0 > 0$ and $\eta^2 \ge 0$.
Then $s_j(\sigma)$ is given by
\bq
\label{dual_delta_I_prescription_covariant}
 s_j\left(\sigma\right)
 & = &
 \sum\limits_{a \in \{j\} \cup \pi} \frac{\eta \cdot k_j}{\eta \cdot k_a}.
\eq
Let us look at an example. 
\begin{figure}
\begin{center}
\includegraphics[scale=1.0]{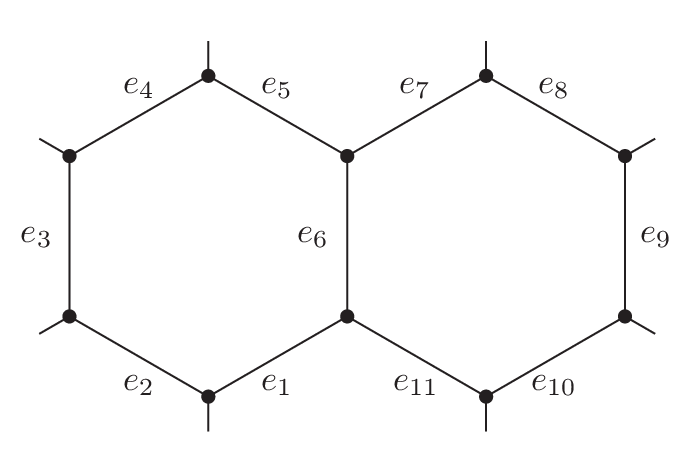}
\end{center}
\caption{
A two-loop eight-point function with $11$ propagators.
}
\label{fig_1}
\end{figure}
Fig.~\ref{fig_1} shows a two-loop eight-point graph $\Gamma$. 
There are $35$ spanning trees, and each spanning tree defines a cut graph.
\begin{figure}
\begin{center}
\includegraphics[scale=1.0]{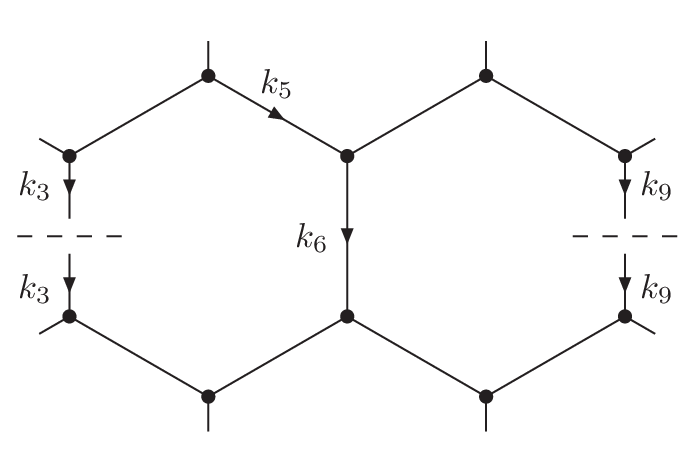}
\end{center}
\caption{
A cut diagram corresponding to $\sigma=\{3,9\}$. We also indicate an orientation for the edges
$e_3$, $e_5$, $e_6$ and $e_9$.
}
\label{fig_2}
\end{figure}
Fig.~\ref{fig_2} shows an example of a cut graph corresponding
to $\sigma=\{3,9\}$. We also indicate in fig.~\ref{fig_2} an orientation for the edges $e_3$, $e_5$, $e_6$ and $e_9$.
As an example we consider $s_5(\sigma)$ and $s_6(\sigma)$. 
The sign of the imaginary part is determined by
\bq
 s_5\left(\sigma\right)
 & = &
 \frac{E_3+E_5}{E_3},
 \nonumber \\
 s_6\left(\sigma\right)
 & = &
 \frac{E_3 E_6+ E_3 E_9 + E_6 E_9}{E_3 E_9}.
\eq
There are ample applications of our result.
Let us give three examples.
First of all, our result is directly geared 
towards numerical methods 
for higher-order computations \cite{Soper:1998ye,Soper:1999xk,Nagy:2003qn,Gong:2008ww,Catani:2008xa,Bierenbaum:2010cy,Bierenbaum:2012th,Buchta:2014dfa,Hernandez-Pinto:2015ysa,Buchta:2015wna,Sborlini:2016gbr,Driencourt-Mangin:2017gop,Driencourt-Mangin:2019aix,Assadsolimani:2009cz,Assadsolimani:2010ka,Becker:2010ng,Becker:2011vg,Becker:2012aq,Becker:2012nk,Becker:2012bi,Goetz:2014lla,Seth:2016hmv,Pittau:2012zd,Pittau:2013qla,Donati:2013voa,Page:2015zca,Gnendiger:2017pys} and paves the way to treat two-loop amplitudes numerically in an efficient and automated way.
Secondly, and in a wider context, it sheds new light on the cancellation of infrared singularities.
Our result allows to discuss the singularity structure of loop integrands in terms of on-shell tree diagrams.
This will be helpful at NNLO and 
beyond \cite{Kosower:2002su,Kosower:2003cz,Weinzierl:2003fx,Weinzierl:2003ra,Gehrmann-DeRidder:2005cm,GehrmannDeRidder:2007jk,Daleo:2009yj,Boughezal:2010mc,Gehrmann:2011wi,Abelof:2011jv,Abelof:2012he,Somogyi:2005xz,Somogyi:2006da,Somogyi:2006db,Aglietti:2008fe,Somogyi:2008fc,Somogyi:2009ri,Bolzoni:2010bt,DelDuca:2016csb,Somogyi:2017bui,Catani:2007vq,Catani:2019iny,Czakon:2010td,Czakon:2011ve,Czakon:2014oma,Gaunt:2015pea,Boughezal:2015dva,Boughezal:2015eha,Magnea:2018hab,Magnea:2018ebr}.
Thirdly and on the formal side, our approach also suggests an extension 
of the concept of scattering forms \cite{Arkani-Hamed:2017tmz,Mizera:2017rqa,delaCruz:2017zqr,Arkani-Hamed:2017mur}
from tree-level towards loops.
This will be explored in a future publication.

\subsection*{Acknowledgements}

This work has been supported by the 
Cluster of Excellence ``Precision Physics, Fundamental Interactions, and Structure of Matter'' 
(PRISMA+ EXC 2118/1) funded by the German Research Foundation (DFG) 
within the German Excellence Strategy (Project ID 39083149).

\subsection*{Note}

In the first version of this article we erroneously assumed that $S_{\sigma \alpha}$ is independent of $\sigma$ and $\alpha$ and given
by $1/2^l$.
This is not correct and has been pointed out in \cite{Capatti:2019ypt}.
In the present version we corrected theorem 1.

\bibliography{/home/stefanw/notes/biblio}
\bibliographystyle{/home/stefanw/latex-style/h-physrev5}

\end{document}